\documentclass[letterpaper, 10 pt, conference]{ieeeconf}  

\IEEEoverridecommandlockouts                              

\title{\LARGE \bf
Dynamic Defender-Attacker Blotto Game
}

\author{Daigo Shishika$^1$, Yue Guan$^2$, Michael Dorothy$^3$, and Vijay Kumar$^4$
\thanks{We gratefully acknowledge the support of ARL grant ARL DCIST CRA W911NF-17-2-0181. The views expressed in this paper are those of the authors and do not reflect the official policy or position of the United States Army, Department of Defense, or the United States Government.}
\thanks{$^1$Daigo Shishika is with the Department of Mechanical Engineering at George Mason University {\tt\small dshishik@gmu.edu}}
\thanks{$^2$Yue Guan is with the School of Aerospace Engineering at Georgia Institute of Technology {\tt\small yguan44@gatech.edu}}
 \thanks{$^3$Michael Dorothy is with CCDC Army Research Laboratory, Adelphi, MD 20783, USA {\tt\small michael.r.dorothy.civ@mail.mil}}
 \thanks{$^4$Vijay Kumar is with GRASP Laboratory at the University of Pennsylvania {\tt\small kumar@seas.upenn.edu}}
 \thanks{*The first two authors contributed equally as co-first authors.}
}

\usepackage{graphicx}
\usepackage{amsmath, bm}
\usepackage{color}
\usepackage{float}
\usepackage{mathtools}
\usepackage[ruled,vlined,linesnumbered]{algorithm2e}

\usepackage{graphics} 
\usepackage{epsfig} 
\usepackage{times} 
\usepackage{amsmath} 
\usepackage{amssymb}  
\usepackage{color}
\usepackage[noend]{algpseudocode}
\usepackage{booktabs,lipsum}
\usepackage{accents}
\usepackage{cite}
\usepackage{amsthm}
\usepackage[dvipsnames]{xcolor}

\newcommand{\x}{\mathbf{x}}
\usepackage{dsfont}

\newcommand{\xreq}{\mathbf{x}^\req}
\newcommand{\xmagreq}{X^\req}
\newcommand\polyreq{\mc{P}^\req}

\newcommand{\vx}{\mathbf{v}}
\newcommand{\vy}{\mathbf{w}}

\newcommand{\y}{\mathbf{y}}

\newcommand{\mc}[1]{\mathcal{#1}}

\newcommand{\mf}[1]{\mathbf{#1}}

\newcommand{\nodeset}{\mc{V}}
\newcommand{\edgeset}{\mc{E}}
\newcommand{\graph}{\mc{G}}

\newcommand{\K}{\mathcal{K}}
\newcommand{\F}{\mathcal{F}}
\newcommand{\hatK}{\hat{K}}
\newcommand{\hatF}{\hat{F}}
\newcommand{\hatcalK}{\hat{\mathcal{K}}}
\newcommand{\hatcalF}{\hat{\mathcal{F}}}

\newcommand{\convexhull}[1]{\mathrm{Conv}\big(#1\big)}

\newtheorem{problem}{Problem}
\newtheorem{definition}{Definition}
\newtheorem{lemma}{Lemma}
\newtheorem{fact}{Fact}
\newtheorem{theorem}{Theorem}

\newtheorem{remark}{Remark}

\def\graph{\mc{G}}
\def\reals{\mathbb{R}}

\def\neighbor{\mc{N}}
\def\inneighbor{\mc{N}^-}
\def\outneighbor{\mc{N}^+}
\def\dout{d^+}
\def\din{d^-}
\def\doutmax{d^+_\mathrm{max}}
\def\doutmin{d^+_\mathrm{min}}

\def\RSet{\mathcal{R}}
\def\SSet{\mathcal{S}}

\def\deltax{\vec{\x}}
\def\deltay{\vec{\y}}

\newcommand{\defender}{the Defender}
\newcommand{\attacker}{the Attacker}

\def\region{\mathcal{R}}

\newcommand{\req}{\text{req}}

\usepackage{yaacro}
\begin{acgroupdef}[list=acronimos]
    \acdef{DAB}{Defender-Attacker Blotto}
    \acdef{dDAB}{dynamic Defender-Attacker Blotto}
    \acdef{dDARA}{dynamic Defender-Attacker Resource Allocation}
\end{acgroupdef}

\newcommand{\thegame}{\ac{dDAB} game}

\begin{document}

\maketitle
\thispagestyle{empty}
\pagestyle{empty}

\begin{abstract}
This work studies a dynamic, adversarial resource allocation problem in environments modeled as graphs. A blue team of defender robots are deployed in the environment to protect the nodes from a red team of attacker robots. We formulate the engagement as a discrete-time dynamic game, where the robots can move at most one hop in each time step. The game terminates with the attacker's win if any location has more attacker robots than defender robots at any time. The goal is to identify dynamic resource allocation strategies, as well as the conditions that determines the winner: graph structure, available resources, and initial conditions. We analyze the problem using reachable sets and show how the outdegree of the underlying graph directly influences the difficulty of the defending task. Furthermore, we provide algorithms that identify sufficiency of attacker's victory.
\end{abstract}



\section{Introduction}
Multi-robot task allocation problems have been studied extensively for various application areas with different levels of abstraction \cite{korsah2013comprehensive, khamis2015multi}.
Of particular interest to this paper are the settings where a large population of robots must move to a set of locations to perform their tasks~\cite{berman2009optimized,prorok2017impact}.
Berman et al.~\cite{berman2009optimized} designed distributed control laws to drive the population of robots to a desired distribution over {graph environments}.
The theory was later extended to accommodate heterogeneous robots and tasks with diverse needs~\cite{prorok2017impact}.
However, the theoretical analysis in these works focused on the steady-state performance, and they did not consider how the system reacts to changing conditions or the presence of adversaries.

With the presence of adversaries, resource allocation problems can be formulated as Colonel Blotto games
\cite{roberson2006colonel,powell2009sequential,chandan2020showing,kovenock2018optimal}.
In the most standard version~\cite{gross1950continuous}, two colonels allocate their resources to multiple locations.
Whoever allocated more resource wins that location, and each colonel aims to maximize the number of locations s/he wins.
Different variants of Colonel Blotto games have been studied, including asymmetric budget~\cite{roberson2006colonel}, asymmetric information~\cite{paarporn2019characterizing}, etc. 
However, most of the formulations consider static games assuming that the desired allocation is achieved instantly and thus ignoring the dynamics that are involved in the resource transportation.\footnote{Although there are works that consider dynamical extensions of Colonel Blotto games \cite{konrad2018budget,hajimirsaadeghi2017dynamic,klumpp2019dynamics}, they do not consider the transportation of the resources in the environment.}






In this work, we take inspirations from the above fields and study resource allocation problem with two main emphases: dynamic reallocation in \emph{time-varying} situation and the presence of \emph{adversarial} resources.
The innovation of our work is in the combination of Colonel Blotto game with the idea of population dynamics over graphs.
Instead of achieving the desired allocation instantly, we require the resources\footnote{We use the terms robots and resources interchangeably. The term ``player'' is reserved for the entity (the Defender or the Attacker) that determines the allocation of those robots / resources.} to \emph{traverse} through the environment.
Now, the resources from competing teams meet during the transition, and their interactions influence the performance of each team.

\begin{figure}[t]
    \centering
    \includegraphics[width=0.4\textwidth]{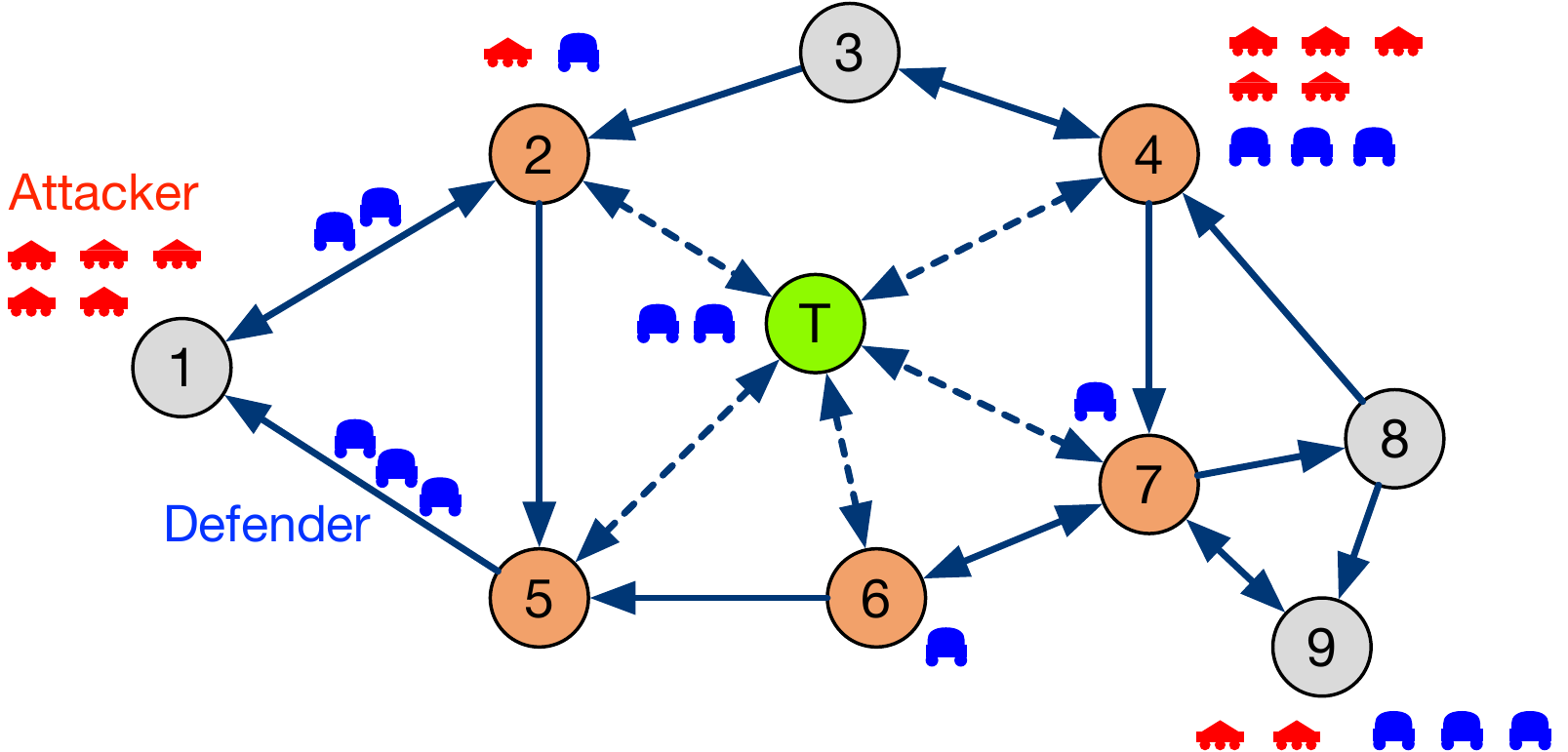}
    \vspace{-5pt}
    \caption{Illustration of the motivating perimeter-defense scenario. The node T (green) denotes the high-value target. The Defender must protect the \emph{entry nodes} (orange) by allocating more defender robots (blue) than the attacker robots (red).}
    \label{fig:illustration}
\end{figure}

As a first step in studying the 
dynamic and adversarial resource allocation problem,
this paper considers a simple version motivated by \emph{perimeter defense} scenario 
\cite{shishika2020ral} 
as depicted in Fig.~\ref{fig:illustration}.
Two players, the \defender{} and the \attacker{}, allocate their resources among the nodes of a directed graph.
The resources can only move along the directed edges, and only one hop is allowed at each time step.
Suppose there is a location with a high-value target that the \defender{} must protect, and that the \attacker{} only has access to that location through a set of ``entry nodes''  (see Fig.~\ref{fig:illustration}).
The Defender must ensure that all the entry nodes have more defender robots than the attacker robots \emph{at all time}.
On the other hand, \attacker{} wins if it ``breaches'' any one of the entry nodes by allocating more robots to it.
Figure~\ref{fig:illustration} presents a situation where node \#4 is breached by the Attacker.

We formulate the scenario as a Game of Kind,
i.e. the performance of a team is solely characterized by whether it wins or not.
In addition, by assuming sequential actions, we exploit the structures in the dynamics of the game and leverage the reachable set approach~\cite{sun2017multiple} to efficiently predict the evolution of the game. 
We will then establish sufficient winning conditions for the players.

The contributions of this paper are: (i) the formulation of a novel resource allocation problem that has high relevance to security applications; (ii) the development of an efficient analysis approach based on reachable sets; and (iii) the derivation of sufficient winning conditions for the players.
Beyond the results presented in this work, the proposed model has a potential for various extensions to study dynamic and adversarial engagement between robots with traversability constraints.


\section{Problem Formulation}
\label{sec:formulation}
The \thegame{} is played between two players: \defender{} and \attacker{}.
The environment is represented as a directed graph $\graph = (\nodeset,\edgeset)$, where the $N$ nodes represent locations, and the directed edges represent the players' traversability between those locations.
We assume that $\graph$ is strongly connected~\cite{west2001introduction,berman2009optimized}, i.e., every node is reachable from any other node.
%
For notational simplicity, we assume that the two players share the same graph, but the analysis of this paper easily extends to the case where two players have different graphs.

A graph adjacency matrix $A \in \reals^{N \times N}$ captures the traversability of the resources as follows:
\begin{equation}
    \left[A\right]_{ij} = \left\{\begin{array}{cl}
         1 \quad& \text{if directed edge $(j,i) \in \edgeset$},   \\
         0 \quad& \text{otherwise}.
    \end{array} \right. 
\end{equation}
We use $\dout_j = \sum_i [A]_{ij}$ to denote the out-degree of node $j$.

The total amount of resources for \defender{} and \attacker{} are denoted by $X\in \reals_{>0}$ and $Y\in \reals_{>0}$, respectively. 
We consider a discrete time problem with continuous resources.\footnote{Such an assumption on the state vector simplifies the analysis \cite{berman2009optimized,prorok2017impact}, however, we will later show that our algorithms accommodate states that take discrete values.}
The state column vector (resource vector), $\x_t$, represents the distribution of \defender's resources across the environment at time step $t$. At each time step $t$, the state vector {lies on a scaled simplex}:
\begin{eqnarray}
    [\x_t]_i &\geq& 0,\; \forall\;i \in \nodeset \\
    \x_t ^\top \mf{1} &=& X, \label{eqn:state-constraint-2}
\end{eqnarray}
where $\mf{1}$ is a vector of ones with an appropriate dimension. 
The state vector $\y_t$ for \attacker{} also satisfies the same conditions above with $X$ replaced by $Y$.
{We use $\Delta_X$ and $\Delta_Y$ to denote the state space of \defender{} and \attacker{}.}

The significant difference from the original Colonel Blotto game is that the \thegame{} is played over multiple time steps, and that the states evolve according to the following dynamics:
\begin{equation}
    \x_{t+1} = K_t \x_t\;\;\;\text{and}\;\;\; \y_{t+1} = F_t \y_t,
    \label{eq:dynamics}
\end{equation}
where $K_t$ and $F_t$ represent the \emph{transition matrices} for \defender{} and \attacker{} respectively. 
These matrices are left stochastic (column sum is unity) and can take nonzero values only at elements where the adjacency matrix has $[A]_{ij}=1$. 
These matrices represent the action / control executed by the players.
For example, an action $K_t$ of \defender{} is admissible at time $t$ if and only if it satisfies the following linear constraints:
\begin{alignat}{2}
    &\sum_i [K_t]_{ij} = 1 \qquad && \forall \; i \in \nodeset 
    \label{eqn:stochastic-matrix-constraint}\\
    &\quad [K_t]_{ij} \geq 0 \qquad && \forall \; i, j \in \nodeset 
    \label{eqn:non-negativity-constraint}\\
    &\quad [K_t]_{ij} = 0 \qquad && \text{if } A_{ij} = 0.
    \label{eqn:adjacency-constraint}
\end{alignat}
{We denote the admissible set for the $K_t$ matrices as $\mathcal{K}$, which}
depends only on the underlying graph $\graph$ and is time-invariant.
The $F_t$ matrix for \attacker{} also satisfies similar constraints.
We denote the admissible set for $F_t$ as $\mathcal{F}$.\footnote{Under the assumption that the two players have the same graph, we have $\mathcal{F} = \K$. 
For consistency, we still use the notations of $\K$ and $\F$ to denote the two action spaces.}

Similar to Colonel Blotto games~\cite{gross1950continuous,roberson2006colonel},
the engagement at each location is modeled solely based on the amount of resources.
More specifically, \defender{} successfully \emph{guards} a location by allocating at least as much resource as \attacker{} did, whereas \attacker{} \emph{breaches} a location by allocating more than what \defender{} did.
For the \thegame{}, \defender{} wants to prevent \attacker{} from breaching any location.\footnote{In relation to the perimeter-defense scenario described in the introduction, this is the case where every node is an entry node.} 
Therefore, we define the terminal condition corresponding to the Attacker's victory as
\begin{equation}
    \exists \;i \in \nodeset \text{ and } t \geq 0, \;\; \text{such that\;\;} [\y_t]_i > [\x_t]_i,
    \label{eq:terminal_condition}
\end{equation}
i.e., the game terminates with \attacker's victory if it breaches at least one location at some time step $t$. The Defender wins the game if it can prevent \attacker{} from winning indefinitely, or over some time horizon $T_f$.

As for the information structure, we assume that the players make decisions in sequence.
Specifically, \defender{} acts first then \attacker{} acts next: i.e., \attacker{} can select its action after observing how \defender{} allocated its resources.
Considering a realistic scenario where the two players make simultaneous actions, our problem is a worst-case scenario for \defender{}.
Importantly, our scenario accommodates state feedback strategies in contrast to systems with constant transition matrices \cite{berman2009optimized,prorok2017impact}.

An instance of \thegame{} is defined by the following parameters: (i)~the available resources $X$ and $Y$, and (ii)~the underlying graph $\graph$.
The initial condition of the game is the initial states of the two players: $\x_0$ and $\y_0$.

\begin{problem}
What are the game parameters and the initial conditions for which \defender{} (resp. \attacker{}) has a strategy to guarantee its win, and what are those strategies?
\end{problem}
We provide our analysis approach based on reachable sets presented in the next section.

\section{Geometry of Reachable Sets}
\label{sec:reachable-set}
Since the dynamics of the two players are symmetric, we focus on the analysis of \defender{}'s reachable sets and its action space $\K$.
There are two major disadvantages working directly with the action space $\K$: 
(i) the higher dimensionality than the state space, i.e. $|\edgeset| \gg |\nodeset|$,
and (ii) the nonuniqueness in the action for a given pair of states: $\x_t$ and $\x_{t+1}$.
To avoid these issues, we directly consider the possible states that \defender{} can reach in the next time step.
\begin{definition}
Given the current state $\x_t$, the reachable set $\RSet(\x_t)$ contains all states that \defender{} can reach at the next time step with an admissible action. 
Formally, the reachable set from $\x_t$ is given by
\begin{equation}
    \RSet(\x_t) = \{ \x \;|\; \exists K\in\K \text{ such that } \x =K \x_t \}.
\end{equation}
\end{definition}

Due to the bilinear dynamics in~\eqref{eq:dynamics}, the reachable set $\RSet(\x_t)$ can be regarded as the image of the action space $\K$ under a linear transformation defined by $\x_t$.
To better understand the properties of the reachable sets, we first examine the structure of the action space.

\subsection{Action Space as a Polytope}
Under the linear constraints in~\eqref{eqn:stochastic-matrix-constraint}--\eqref{eqn:adjacency-constraint}, the set of admissible actions, $\mathcal{K}$, is a bounded polytope in $|\edgeset|$-dimensional space.
We use the extreme points (vertices) of this polytope to characterize $\mathcal{K}$.

Given the admissible action space $\K$, we define the set of \textit{extreme actions} as
\begin{equation}\label{eqn:extreme-actions}
    \hat{\mathcal{K}} = \big\{K \in \mathcal{K} \text{ such that } [K]_{ij} \in \{0,1\} \big\}.
\end{equation}
In words, $\hatcalK$ contains all admissible actions $K$ whose entries only take value of either 0 or 1. 
It can be shown that the cardinality of $\hatcalK$ is given by $\big\vert \hatcalK \big\vert = \prod_{j \in \nodeset} \dout_j$, where $\dout_j$ is the out-degree of node $j$ in the graph $\graph$.
We use $\ell$ to index the extreme actions in $\hatcalK$, 
i.e. $\hatcalK = \{\hatK^{(\ell)}\}_{\ell=1}^{|\hatcalK|}$.
The following theorem reveals the connection between the extreme actions and the admissible action set.
\begin{theorem}\label{thm:extreme-action}
The extreme actions defined in~\eqref{eqn:extreme-actions} are the vertices of polytope $\mathcal{K}$. 
Formally,
\begin{equation}\label{eqn:convex-hull}
    \K = \convexhull{\hatcalK}.
\end{equation}
Consequently, for any admissible action $K \in \K$, 
there is a set of non-negative coefficients $\bm{\lambda}=\{\lambda^{(\ell)}\}_{\ell=1}^{|\hatcalK|}$ such that $\sum_{\ell=0}^{|\hatcalK|} \lambda^{(\ell)}=1$ and 
\begin{equation}\label{eqn:action-convex-comb}
    K = \sum_{\ell=0}^{|\hatcalK|} \lambda^{(\ell)} \hatK^{(\ell)}.
\end{equation}
\end{theorem}
\begin{proof}
We provide a proof by double inclusion.
The direction of $\convexhull{\hatcalK} \subseteq \K$ is easy to show, as the extreme actions are all admissible actions and the linear constraints in \eqref{eqn:stochastic-matrix-constraint}--\eqref{eqn:adjacency-constraint} hold under convex combinations.

To show that $\K \subseteq \convexhull{\hatcalK}$, we provide a formula of $\{\lambda^{(\ell)}\}_\ell$ in~\eqref{eqn:action-convex-comb} for an arbitrary $K\in \K$. 
We first define the active edge set $\mathcal{I}^{(\ell)}$ for the extreme action $\hatK^{(\ell)} \in \hatcalK$ as 
\begin{equation*}
    \mathcal{I}^{(\ell)} = \left\{(j,i) \Big\vert [\hatK^{(\ell)}]_{ij}=1\right\}.
\end{equation*}
Then, given any admissible action $K \in \K$, the  coefficients $\lambda^{(\ell)}$ corresponding to the extreme action $\hatK^{(\ell)}$ can be computed as
\begin{equation}
     \lambda^{(\ell)} = \prod_{(j,i)\in \mathcal{I}^{(\ell)}} [K]_{ij}.
\end{equation}
One can further verify that the above formula satisfies~\eqref{eqn:action-convex-comb} and $\sum_{\ell = 1}^{|\hatcalK|} \lambda^{(\ell)} = 1$. 
Consequently, any admissible action is in the convex hull of the extreme actions.
With the double inclusion, we have proved the relation in~\eqref{eqn:convex-hull}. 
\end{proof}

\begin{remark}
The extreme action set $\hatcalK$ depends only on the graph $\graph$, and it only needs to be constructed once.
\end{remark}


The extreme action set for the \attacker{} is denoted as $\hatcalF$ and is defined similarly. 
We use $\{\hatF^{(r)}\}_{r=1}^{|\hatcalF|}$ to index the elements within $\hatcalF$.

\subsection{Reachable Sets as Polytopes}
The reachable sets $\RSet(\x_t)$ are in fact also polytopes in $\Delta_X$ (i.e., $|\nodeset|$-dimensional state space),
and it can be obtained by transforming the action space $\K$ as follows.

\begin{figure}[t]
    \centering
    \includegraphics[width=0.48\textwidth]{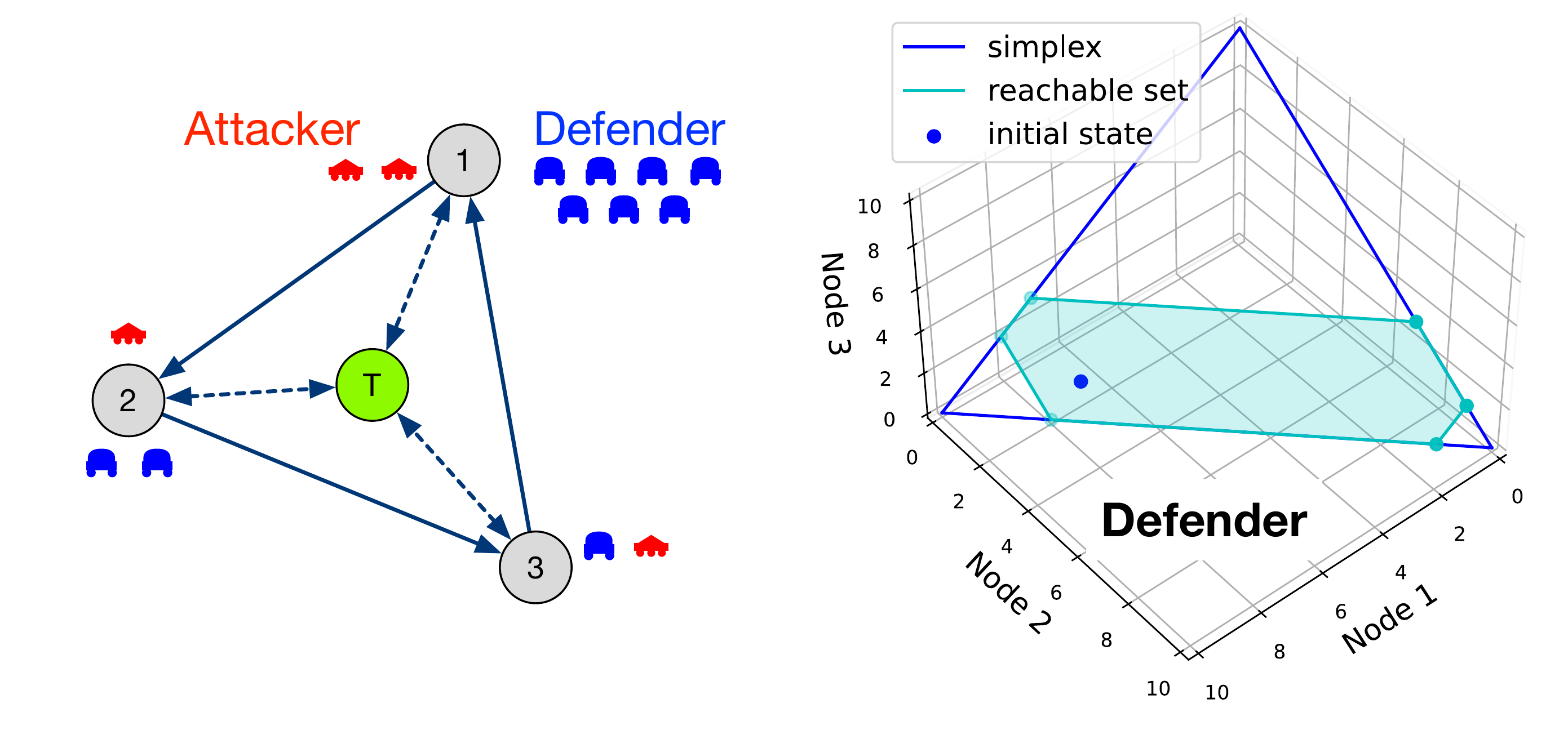}
    \caption{Illustration of reachable set. (a) Directed graph with three nodes, where each node $i$ has a self-loop, i.e. $(i,i) \in \edgeset$ for $i=1,2,3$ (omitted for clarity). (b) Defender's reachable set and its vertices.
    }
    \label{fig:demo-reachable-set}
\end{figure}

For any $\x \in \RSet(\x_t)$, by definition, there is an action $K_t \in \K$, such that $\x = K_t \x_t$.
Based on the characterization of $\K$ in~\eqref{eqn:action-convex-comb}, this $\x$ can be represented as the following convex combination for some $\bm{\lambda}$:
\begin{equation*}
    \x = K \x_t =\bigg(\sum_{\ell=0}^{|\hatcalK|} \lambda^{(\ell)} \hatK^{(\ell)} \bigg)\x_t = \sum_{\ell=0}^{|\hatcalK|} \lambda^{(\ell)} \left(\hatK^{(\ell)} \x_t\right). 
\end{equation*}
Define $\vx^{(\ell)}_{t+1} = \hatK^{(\ell)} \x_t$ to be the state achieved by propagating $\x_t$ with the extreme action $\hatK^{(\ell)}$.
Then, the convex hull of these vertices gives us the polytope $\RSet(\x_t) = \convexhull{\{\vx_{t+1}^{(\ell)}\}_{\ell=1}^{|\hatcalK|}}$, which describes the set of states that \defender{} at $\x_t$ can achieve at the next time step. 
Figure~\ref{fig:demo-reachable-set} presents an example of the reachable set for a three node graph.
Note that discrete resources (robots) are shown in Fig.~\ref{fig:demo-reachable-set}(a), however, the reachable set is given for a continuous state space.



\begin{remark}
Some of the projected points $\vx^{(\ell)}_{t+1}$ may be in the interior of $\RSet(\x_t)$.
One can remove these redundant points through convex hull algorithms such as Quickhull~\cite{barber1996quickhull}. 
In the sequel, we use $\vx^{(\ell)}_{t+1}$ to represent the vertices of $\RSet(\x_t)$. 
\end{remark}


With the same argument for \attacker{}, we can obtain the polytope, $\RSet(\y_t)$, and its vertices $\vy^{(r)}_{t+1} = \hatF^{(r)} \y_t$ for $r=1,2,...,|\hatF|$.
Note that by construction we have the following:  $\mf{1}^\top \vx^{(\ell)}_{t+1} = X$ and $\mf{1}^\top \vy^{(r)}_{t+1}  = Y$ for all these vertices.

Since any state in $\RSet(\x_t)$ can be reached at the next time step from $\x_t$, we will view this polytope as the action space for   \defender{} at state $\x_t$.
This definition of the action space resolves the two issues raised at the beginning of this section (dimensionality and nonuniqueness).

\section{Terminal States} \label{sec:terminal-states}
In this section, we will identify the set of states $(\x_t,\y_t)$ that will immediately lead to a termination in the next time step: i.e., for any Defender action $\x_{t+1}\in\RSet(\x_t)$,   \attacker{} has a strategy $\y_{t+1}\in\RSet(\y_t)$ to satisfy \eqref{eq:terminal_condition} and win at least one location at the next time step.


\subsection{Minimum Resource to Defend}
For \defender{} to defend every location, it is necessary and sufficient if the resource vector, $\x$, matches or outnumbers $\y$ at every node $i$: 
\begin{equation}
    [\x]_i\geq[\y]_i \quad \forall i \in \nodeset.
    \label{eq:vert_condition}
\end{equation}
%
Since \attacker{} takes its action after observing \defender{}'s action, the question becomes, whether there exists $\x_{t+1}\in \RSet(\x_t)$ such that \eqref{eq:vert_condition} is true for all $\y_{t+1} \in \RSet(\y_t)$.
We first rewrite \eqref{eq:vert_condition} into the following form:
\begin{equation*}
    [\x_{t+1}]_i \geq \max_{\y_{t+1} \in \RSet(\y_t)} [\y_{t+1}]_i \quad  \forall i \in \nodeset.
\end{equation*}
Since $ \RSet(\y_t)$ is a bounded polytope, for each node $i$ the optimization  $\max_{\y_{t+1} \in \RSet(\y_t)} [\y_{t+1}]_i$ can be viewed as a linear program, 
whose optimum is attained on one of the vertices of $\RSet(\y_t)$. 
Consequently, we define the minimum required resource at $t+1$ as $\xreq_{t+1}$, whose elements are given by
\begin{equation}\label{eqn:x_req}
    [\xreq_{t+1}]_i =\max_r \left[ \vy_{t+1}^{(r)} \right]_i,
\end{equation}
where $\big\{\vy^{(r)}_{t+1}\big\}_r = \big\{\hatF^{(r)} \y_t\big\}_r$ are the vertices of $\RSet(\y_t)$. 
\begin{remark}
Defender's minimum required resource at the \emph{next} time step, $\xreq_{t+1}=\xreq_{t+1}(\y_t)$, is a function of \attacker{}'s \emph{current} state, $\y_t$.
\end{remark}

By allocating at least $[\xreq_{t+1}]_i$ to node $i$,   \defender{} ensures that this node is defended against all feasible Attacker  actions at time step $t+1$.
If \defender{} allocates $[\x_{t+1}]_i<[\xreq_{t+1}]_i$, then after observing Defender's allocation, \attacker{} has a strategy $\y_{t+1}\in\RSet(\y_t)$ to win location $i$.
Thus $\xreq_{t+1}$ is the necessary and sufficient resources for \defender{} to defend all locations at time $t+1$, given the current $\y_t$.

Notice that $\xmagreq_{t+1}=\mf{1}^\top \xreq_{t+1}$ depends on $\graph$ and $\y_{t}$.
It is easy to see that \defender{} does not have a strategy to guarantee defense if $\xmagreq_{t+1}>X$.
On the other hand, if $\xmagreq_{t+1}\leq X$, \defender{} can guarantee defense by selecting any $\x_{t+1}$ that is inside the polytope $\polyreq(\y_t)$, which we call the \emph{required set}, defined as
\begin{equation}\label{eqn:poly_req}
    \polyreq(\y_t) \triangleq \{ \x_{t+1} \;|\; [\x_{t+1}]_i\geq [\xreq_{t+1}(\y_t)]_i,\; \forall\; i\in \nodeset \}.
\end{equation}
Such a selection is only possible if $ \RSet(\x_t)\cap\polyreq(\y_t)\neq \emptyset$.
Consequently, we define the {\emph{safe set}} for \defender{} as 
\begin{equation}
    \SSet(\x_t,\y_t) = \RSet(\x_t)\cap\polyreq(\y_t).
\end{equation}

Figure~\ref{fig:demo-required-resource} demonstrates the relationship between the reachable set, the required set and the safe set under the three-node example in Fig.~\ref{fig:demo-reachable-set}. 

\begin{figure}[t]
    \centering
    \includegraphics[width=0.48\textwidth]{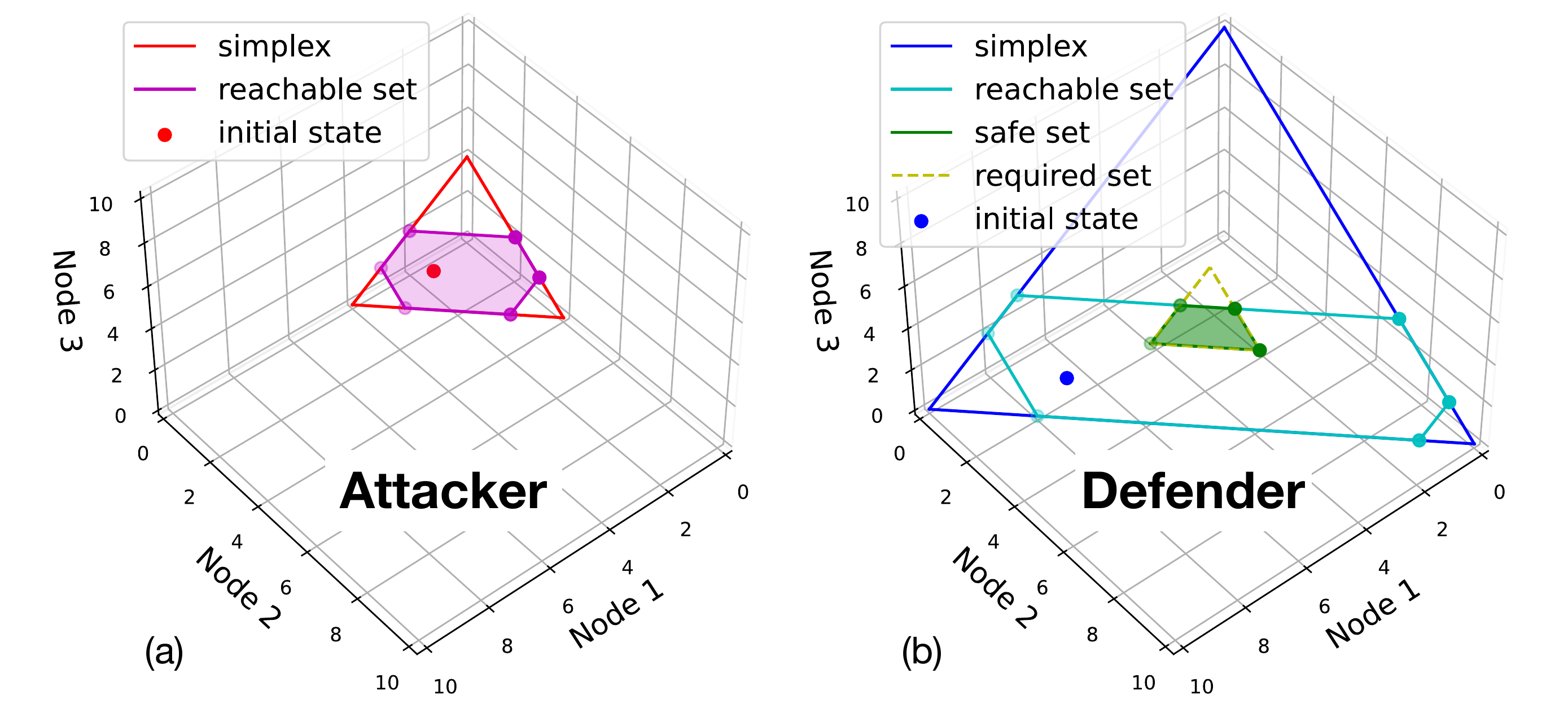}
    \caption{The required set $\polyreq(\y_t)$ resulted from \attacker{}'s reachable set and the safe set as the intersection of $\polyreq(\y_t)$ and $\RSet(\x_t)$.}
    \label{fig:demo-required-resource}
\end{figure}

\begin{lemma}\label{lem:xreq}
The game is one step away from termination, i.e., \attacker{} has a strategy to win at time $t+1$ if the safe set is empty at time $t$: $\SSet(\x_t,\y_t) = \emptyset$.
\end{lemma}
\begin{proof}
For any admissible Defender strategy $\x\in\RSet(\x_t)$, the condition $\SSet(\x_t,\y_t) = \emptyset$ states that $\x \notin \polyreq(\y_t)$.
Consequently, there exists a node $j$ such that $[\x]_j<[\xreq_{t+1}]_j$.
By construction, there is a vertex with index $r$ such that $[\vy_{t+1}^{(r)}]_j=[\xreq_{t+1}]_j$.
Therefore, \attacker{} can allocate according to $\y_{t+1}=\vy_{t+1}^{(r)} \in \RSet(\y_t)$ and breach node  $j$.
\end{proof}

\subsection{Degenerate Parameter Regime}
We studied the minimum required resource $\xreq$ using the vertices of the \attacker{}'s reachable set.
Here we consider a graph-theoretic interpretation.
The maximum possible resource to be allocated to node $i$ at time $t+1$ is the sum of resources at all neighboring locations with directed edges pointing to $i$ at time $t$.
We use $\inneighbor_i$ to denote the \emph{in-neighbors} of $i$, which may include $i$ itself if there is a \emph{self-loop}: $(i,i)\in\edgeset$.
An alternative definition of $\xreq$ is then
\begin{equation}
    [\xreq_{t+1}]_i = \sum_{j\in\inneighbor_i} [\y_t]_j.
\end{equation}

Recalling that $\dout_i$ denotes the \emph{outdegree} of node $i$,
the required total resource $\xmagreq$ satisfies
\begin{eqnarray}
    \xmagreq_{t+1} &=& \sum_i^N [\xreq_{t+1}]_i = \sum_i^N\sum_{j\in\inneighbor_i} [\y_t]_j \label{eq:xreq_count1}\\
    &=& \sum_j^N\sum_{i\in\outneighbor_j} [\y_t]_j =  \sum_j^N \dout_j [\y_t]_j. \label{eq:xreq_count2}
\end{eqnarray}
Note that from the first to the second line, we convert the counting based on inward edges to the one with outward edges.
Using this expression, we obtain the following bound:
\begin{equation}
    \doutmin  Y \leq \xmagreq \leq \doutmax Y,
    \label{eq:xreq_bound}
\end{equation}
where $\doutmin= \min_i \dout_i$ and $\doutmax= \max_i \dout_i$ are the minimum and maximum outdegree of $\graph$.
The time subscript is omitted since the bound is valid for all $t$.
The first (resp. second) equality holds when all the resource $Y$ is concentrated at $i$ with $\dout_i=\doutmin$ (resp. $\dout_i=\doutmax$).

\begin{remark}
Note that if $\dout_i=\dout$, $\forall i\in\nodeset$, then we have $\xmagreq = \dout Y$.
This means that $\xmagreq$ is fixed for any current Attacker state $\y_t$.
\end{remark}

The following result provides a parameter regime where the game trivially ends with \attacker{}'s win.

\begin{theorem}[Degenerate Attacker Winning Game]
Suppose there exists a node with self-loop $(i,i) \in \edgeset$,
and the total resources satisfy 
\begin{equation}
    X< \doutmax Y,
\end{equation}
then the Attacker can win the game from any initial state.
\end{theorem}
\begin{proof}
Due to the existence of a self-loop $(i,i)\in \edgeset$,   \attacker{} can concentrate all its resource to the location $i$ in finite time. 
Furthermore, under the assumption of strongly connected graph, \attacker{} can then move its concentrated resource from node $i$ to node $i^*$ with $\dout_{i^*}=\doutmax$ in finite time.\footnote{Both procedures take at most $T_D$ time steps, where $T_D$ is the diameter of the graph.}
Once this configuration is achieved, we have $\xmagreq=\doutmax Y>X$. Therefore,  \defender{} does not have a strategy to defend every node. That is, \attacker{} has a strategy to win.
\end{proof}




\section{Strategies and Winning Conditions} \label{sec:winning-conditions}
From the results of the previous section, we will focus our attention to the case where the total resources satisfy
\begin{equation}
    \doutmax Y \leq X .
    \label{eq:population_nontrivial_game}
\end{equation}
\subsection{Defender Winning Scenario}
This section identifies conditions in which \defender{} has a strategy to win the game by protecting all the locations indefinitely.
In general the following two conditions guarantee that \defender{} can win the game.
\begin{itemize}\setlength{\itemindent}{1.5em}
    \item[(\textbf{C1a})] Initial condition:
    \begin{equation}
        \SSet(\x_0,\y_0)\neq \emptyset.
    \end{equation}
    \item[(\textbf{C1b})] Inductive condition: \\
    For any $\x_t$ and $\y_t$ for which 
    $\SSet(\x_t,\y_t)\neq \emptyset$,  
    \begin{eqnarray}
        &\exists\; \x_{t+1}\in\RSet(\x_t)\text{\;such that\;}\\ 
        &\SSet(\x_{t+1},\y_{t+1})\neq \emptyset,\; \forall\;\y_{t+1}\in\RSet(\y_{t+1}).
    \end{eqnarray}
\end{itemize}
The condition (C1a) ensures that \defender{} can survive the first time step. 
The inductive condition (C1b) states that if \defender{} is able to survive the current time step, then it can also survive the next time step.

\begin{lemma}\label{lem:guarding_condition}
The necessary and sufficient condition for guarding, $\x_t\in \polyreq(\y_{t-1})$, is equivalent to 
\begin{equation}\label{eq:guarding_condition}
    [\x_t]_i \geq \sum_{j\in\neighbor_i^-}[\y_{t-1}]_j,\;\forall\; i\in\nodeset.
\end{equation}
\end{lemma}
\begin{proof}
The maximum resource that \attacker{} can allocate to node $i$ at time $t$ is the right-hand side of \eqref{eq:guarding_condition}: i.e.,
\begin{equation}
    [\y_t]_i \leq  \sum_{j\in\neighbor_i^-}[\y_{t-1}]_j,\;\forall\; i\in\nodeset.
\end{equation}
The equality holds when the neighboring nodes $\neighbor_i^-$ send all their resources to $i$.
The statement of the lemma trivially holds from this result.
\end{proof}

The following theorem applies the conditions (C1a) and (C1b) to a specific class of graphs.

\begin{theorem}[Cycle Graph]\label{thm:cycle_graph}
Suppose $\graph$ is a directed cycle graph with self-loops ($\dout_i=\din_i=2$, and $(i,i)\in\edgeset$, $\forall\;i\in \nodeset$), and the resources satisfy
    $X \geq 2Y,$
then \defender{} has a strategy to win the game iff 
$\SSet(\x_0,\y_0)\neq\emptyset$. 
\end{theorem}

\begin{proof}
Necessity is straightforward based on the construction of $\polyreq(\y_0)$ and Lemma~\ref{lem:xreq}.
We show the sufficiency in the following.

Let us label the nodes of the cycle graph so that $[A]_{i,i-1}=1$: i.e., there is an outgoing edge from $i-1$ to $i$ (with the convention $i-1=N$ for $i=1$). Suppose we have $\x_t$ and $\y_t$ that satisfy $\SSet(\x_t,\y_t)\neq \emptyset$.
As a special case of Lemma~\ref{lem:guarding_condition} with $\neighbor_i^-=\{i-1,i\}$, any strategy $\x_{t+1}\in\SSet(\x_t,\y_t)$ satisfies
\begin{equation}
    [\x_{t+1}]_i \geq [\y_t]_i + [\y_t]_{i-1}.
    \label{eq:ring_x}
\end{equation}

Let $[\deltay_t]_i$ denote the amount of resource that moves from $i$ to $i+1$ at $t$ based on the action $F_t$.
Note that $[\deltay_t]_i$ is in the interval $0\leq [\deltay_t]_i \leq [\y_t]_i$.
Then we have the following:
\begin{equation}
    [\y_{t+1}]_i = [\y_{t}]_i - [\deltay_{t}]_i + [\deltay_{t}]_{i-1},
\end{equation}
which gives the necessary and sufficient condition for a Defender's action, $[\x_{t+2}]_i$, to be \emph{safe} as follows:
\begin{eqnarray}\label{eq:condition_x_t2}
    [\x_{t+2}]_i &\geq& [\y_{t+1}]_i + [\y_{t+1}]_{i-1} \notag \\
    &=& [\y_{t}]_i - [\deltay_{t}]_i + [\deltay_{t}]_{i-1} \notag \\
    &&+ [\y_{t}]_{i-1} - [\deltay_{t}]_{i-1} + [\deltay_{t}]_{i-2} \notag \\
    &=& \left( [\y_{t}]_i + [\y_{t}]_{i-1}\right) - ([\deltay_{t}]_i+[\deltay_{t}]_{i-1}) \notag\\
    &&+ ([\deltay_{t}]_{i-1}+[\deltay_{t}]_{i-2}).
\end{eqnarray}

Now, consider the following candidate strategy for \defender{}, $\deltax_{t+1}$, described by the amount of resource that moves out of each node $i$:
\begin{equation}
    [\deltax_{t+1}]_i = [\deltay_{t}]_i+[\deltay_{t}]_{i-1}.
\end{equation}
Note that this action is feasible since \eqref{eq:ring_x} implies $[\x_{t+1}]_i \geq [\deltay_{t}]_i+[\deltay_{t}]_{i-1}$: i.e., node $i$ can ``afford'' this much resource to be sent to $i+1$.
With this action, we have
\begin{eqnarray}
    [\x_{t+2}]_i &=& [\x_{t+1}]_i - [\deltax_{t+1}]_i + [\deltax_{t+1}]_{i-1} \notag \\
    &\geq& \left( [\y_{t}]_i + [\y_{t}]_{i-1}\right)- [\deltax_{t+1}]_i + [\deltax_{t+1}]_{i-1}, \quad
\end{eqnarray}
which reduces to the necessary and sufficient condition \eqref{eq:condition_x_t2}.
From Lemma~\ref{lem:guarding_condition}, we have $\x_{t+2}\in\polyreq(\y_{t+1})$. 
Also, the feasibility of the above action gives us $\x_{t+2}\in\RSet(\x_{t+1})$.
Hence, our candidate strategy provides the condition $\SSet(\x_{t+1},\y_{t+1})\neq\emptyset$.
We have shown that (C1b) is satisfied, which completes the sufficiency proof.
\end{proof}

The result of Theorem~\ref{thm:cycle_graph} shows that for cycle graphs, \attacker{} can win if and only if it can win immediately at $t=1$.
If \defender{} survives one time step, then \attacker{} can never win.
This makes it easy for us to find the winning region for \defender{}, $\region_D$, and the winning region for \attacker{}, $\region_A$: i.e., for cycle graphs we have
\begin{eqnarray}
    \region_D &=& \{[\x,\y]\;|\;\SSet(\x,\y)\neq\emptyset\},\\
    \region_A &=& \{[\x,\y]\;|\;\SSet(\x,\y)=\emptyset\}.
\end{eqnarray}
Note that the winning regions are not simple for a general graph since \attacker{} may be able to win after multiple time steps.
Strategies and winning regions for more complex graphs are subjects of ongoing work.

\subsection{Attacker Winning Scenario}
To analyze the winning condition for \attacker{} over multiple time steps, we define the notion of state-to-state (S-S) time.
\begin{definition}
Given two configurations $\x_s$ and $\x_g$, the state-to-state (S-S) time for \defender{}, $\tau^X(\x_s,\x_g)$, is defined as the minimum time steps required to achieve $\x_g$, starting from $\x_s$.
The S-S time for \attacker{}, $\tau^Y$, is defined similarly.
\end{definition}

Using this notion, we provide another sufficient condition for \attacker{} to win the game.

\begin{itemize}\setlength{\itemindent}{1em}
    \item[(\textbf{C2})] There exists a pair $(\y_s,\y_g)$ such that for all $\x_{s}\in\polyreq(\y_s)$ and $\x_{g}\in\polyreq(\y_g)$, we have
    \begin{equation}
        \tau^X(\x_s,\x_g) > \tau^Y(\y_s,\y_g).
    \end{equation}
\end{itemize}
If the condition (C2) is true, and if \attacker{} is able to reach $\y_s$ in finite time, then it can transition from $\y_s$ to $\y_g$ in minimum time to ensure its victory.


%

To formalize the idea, we extend the definition of the reachable set $\RSet(\x_t)$ of a single point $\x_t$ to the reachable set of a polytope, based on which we propagate the reachable set over multiple time steps. 
We first consider the set of all possible states that can be reached in the next time step from a polytope $D_t \in \Delta_X$:
\begin{equation}\label{eqn:rechability-set}
    \RSet(D_t) = \left\{ \x =  K\x_t  \; \vert\; K\in \K \text{ and } \x_t \in D_t  \right\}.
\end{equation}

The set $\RSet(D_t)$ can be constructed as a union of the reachable sets of all points $\x_t \in D_t$ via
$\RSet(D_t) = \bigcup_{\x_t \in D_t} \RSet(\x_t)$.
Suppose that the vertices of the polytope $D_t$ are $\{\x_{D_t}^{(s)}\}_s$, then any point, $\x_{t+1}$, in $\RSet(D_t)$ can be expressed as:
\begin{eqnarray}
    \label{eqn:characterization-reachable-D}
    \x_{t+1} &=&  \left(\sum_{\ell}\lambda^{(\ell)}\hatK^{(\ell)}\right)  \left(\sum_s \theta^{(s)} \x_{D_t}^{(s)} \right)\\
    &=& \sum_{\ell, s} \left(\theta^{(s)} \lambda^{(\ell)}\right) \left(\hatK^{(\ell)}\x_{D_t}^{(s)}\right),    
\end{eqnarray}
where $\{\theta^{(s)}\}_s$ and $\{\lambda^{(\ell)}\}_{\ell}$ are convex coefficients. 

\begin{fact}
The set $\RSet(D_t)$ is the convex hull of the vertices of $D_t$ propagated with extreme actions: 
\begin{equation}
\label{eqn:multi-step-RSet}
    \RSet(D_t)= \convexhull{\left\{\hatK^{(\ell)}\x_{D_t}^{(s)}\right\}_{\ell,s}}.
\end{equation}
\end{fact}

With the above extension, the mapping $\RSet(\cdot)$ can be cascaded $\tau$ times to obtain the reachability set after $\tau$ time steps.
The following algorithm shows how the S-S can be computed for \defender{}. 
(The one for \attacker{} can be computed similarly.)
{In the sequel, we use $D_t$ to denote the states that can be reached at time $t$.}

\begin{algorithm}[h]
\SetAlgoLined
\SetKwInput{KwInputs}{Inputs}
\KwInputs{Traversability $A$; start and goal states $\x_s,\x_g$;}
Compute the extreme action sets $\hatcalK$ via~\eqref{eqn:extreme-actions}\;
Set $D_0 \leftarrow \x_s$, and $\tau \leftarrow 0$\;
 \While{$\x_g \notin D_\tau$}{
  Propagate:
  $D_{\tau+1} \leftarrow \RSet(D_\tau)$ via \eqref{eqn:multi-step-RSet}\;
  $\tau\leftarrow \tau+1$\;
 }
 \textbf{Return} $\tau$
 \caption{Finding S-S Time}
 \label{alg:S-S-time}
\end{algorithm}
\vspace{-10pt}

Algorithm~\ref{alg:attacker-winning-graph} shows a sampling-based approach\footnote{The samples can be generated randomly or through meshing $\Delta_Y$.}
to find a pair of start and goal configurations $(\y_s,\y_g)$ that satisfies (C2).
Note that $\RSet^\tau(\cdot)$ denotes the cascade of propagation $\RSet(\cdot)$ applied $\tau$ times: e.g., $\RSet^2(\cdot)=\RSet(\RSet(\cdot))$.

\begin{algorithm}[h]
\SetAlgoLined
\SetKwInput{KwInputs}{Inputs}
\KwInputs{Traversability $A$; available resources $X$, $Y$\;}
Compute the extreme action sets $\hatcalK, \hatcalF$ via~\eqref{eqn:extreme-actions}\;
Sample a set of \attacker{} configuration pairs $\mc{Y}_{c}$  \;
 \For{$(\y_i,\y_j) \in \mc{Y}_c$}{
  $\tau\gets\tau^Y(\y_i,\y_j)$ using Algorithm~1\;
  Find initial required set: $D_1\gets \polyreq(\y_i) $\;
  Propagate safe set: $D_{\tau+1} \leftarrow \RSet^\tau (D_1)$\;

  \If{$D_{\tau+1} \cap \polyreq(\y_j) = \emptyset$}{
  \textbf{Return:} $(\y_s,\y_g) \leftarrow (\y_i,\y_j)$
  }}
 \caption{Attacker Winning Graph}
 \label{alg:attacker-winning-graph}
\end{algorithm}
\vspace{-10pt}

If Algorithm~\ref{alg:attacker-winning-graph} returns a solution,
then \attacker{} has an open-loop strategy to win the game regardless of \defender{}'s strategy.
Specifically, if \attacker{} can achieve $\y_s$ in finite time steps,\footnote{This is true for any initial state, $\y_0$, as long as $\graph$ is strongly connected and has at least one self-loop.} then from there it can reach $\y_g$ in $\tau^Y(\y_s,\y_g)$ time steps.
Let $t=0$ denote the time at which \attacker{} reaches the state $\y_s$.
For \defender{} to survive, it is necessary to have $\x_1\in\polyreq(\y_s)=D_1$ and $\x_{\tau+1}\in\polyreq(\y_g)$.\footnote{Note that this is necessary, but not a sufficient condition since \defender{} also must guard all the time steps in between: $1<t<\tau+1$.}
However, if $D_{\tau+1}\cap \polyreq(\y_g)=\emptyset$, then \defender{} does not have a strategy to satisfy $\x_{\tau+1}\in\polyreq(\y_g)$, and thus \attacker{} has a guarantee to win.

The result of Algorithm~\ref{alg:attacker-winning-graph} is a {characterization} of the graph $\graph$, and is completely independent of the initial conditions.
Even if a pair that satisfies condition (C2) does not exist, which Algorithm~\ref{alg:attacker-winning-graph} checks, \attacker{} may still win if the game starts in a favorable initial configuration.

To take the initial states into account, we present Algorithm~\ref{alg:open-loop-attacker} as a complement to Algorithm~\ref{alg:attacker-winning-graph}.
Given the initial states $(\x_0,\y_0)$, Algorithm~\ref{alg:open-loop-attacker} provides a sampling-based approach to find an open-loop strategy for the Attacker to win the game within $t_\text{max}$ time steps.

\begin{algorithm}[h]
\SetAlgoLined
\SetKwInput{KwInputs}{Inputs}
\KwInputs{Traversability $A$; initial states $\x_0, \y_0$\;}
Compute the extreme action sets $\hatcalK, \hatcalF$ via~\eqref{eqn:extreme-actions}\;
Initialize $D_1^X \leftarrow \RSet(\x_0)$ and $D_0^Y \leftarrow \{\y_0\}$\;
Initialize $t \leftarrow 1$\;
 \While{$t < t_\text{max}$}{

  Sample a set of goal points, $\mc{Y}_{c}$, from $D^Y_{t-1}$\;
  
  \For{$\y_{c} \in \mc{Y}_{c}$}{
  \If{$D^X_t \cap \polyreq(\y_{c}) = \emptyset$}{
  \textbf{Return:} $\y_g \leftarrow \y_{c}$
  }}
  Defender's reachable set: $D^X_{t+1} \leftarrow \RSet(D^X_t)$\;
  Attacker's reachable set: $D^Y_t \leftarrow \RSet(D^Y_{t-1})$\;
  
  $t \leftarrow t+1$\;
 }
 \caption{Open-loop Attacker Strategy}
 \label{alg:open-loop-attacker}
\end{algorithm}
\vspace{-10pt}

\begin{theorem}[Open-loop Attacker Strategy]
If Algorithm~\ref{alg:open-loop-attacker} returns a solution $\y_g$, then regardless of \defender{}'s strategy, \attacker{} has an open-loop strategy to win the game: i.e., reach a state with guaranteed win. 
\end{theorem}
\begin{proof}
Suppose the algorithm terminates at $t=\tau+1$.
By construction, $D_{\tau+1}^X$ contains all points that \defender{} can reach within $\tau$ time steps from any point $\x_1\in D^X_1=\RSet(\x_0)$. 
The condition $D^X_{\tau+1} \cap \polyreq(\y_{g}) = \emptyset$ in line 7 implies that for all $\x_1\in\RSet(\x_0)$ and $\x_g \in \polyreq(\y_{g})$, we have
\begin{equation}
    \tau^X(\x_1,\x_g) > \tau \geq \tau^Y(\y_0,\y_g),
\end{equation}
which is equivalent to (C2) with an additional constraint that $\x_s \in \RSet(\x_0)$.
By construction, \attacker{} can achieve $\y_g$ in $\tau \geq \tau^Y(\y_0,\y_g)$ time steps, and it will win in the next time step (after observing \defender{}'s action at $\tau+1$).\footnote{The actual strategy (which may be non-unique) can be found by a linear-program based strategy extraction algorithm, which we will present in future publications due to space limit.}
Note that the above argument accommodates the degenerate case where $\y_g=\y_0$ and $\tau=0$.
\end{proof}

\begin{remark}[Discrete state space]
Although the reachability analysis assumes $\x$ and $\y$ to take continuous values, {one may restrict Algorithms~\ref{alg:attacker-winning-graph} and \ref{alg:open-loop-attacker} to sample integer-valued resource set $\mc{Y}_c$ to accommodate the discrete state space cases.}
\end{remark}

\section{Illustrative Examples}\label{sec:examples}
This section provides examples that illustrate the theoretical results presented in the previous sections.


\subsection{Attacker Winning Graph}
We first demonstrate the idea of S-S time and Algorithm~\ref{alg:attacker-winning-graph}.
Figure~\ref{fig:ex1_attacker} shows a graph with $\doutmax=3$;
every node has a self-loop, but is omitted for clarity.
Suppose $Y=1$ and $X=3$. 
This amount of resources ensures that the game satisfies~\eqref{eq:population_nontrivial_game} and is non-degenerate.
For this graph and resources, Algorithm~\ref{alg:attacker-winning-graph} finds a solution and outputs $\y_s = [1,0,0,0,0]^\top$ and $\y_g = [0,1,0,0,0]^\top$ as depicted by the red robots in Fig.~\ref{fig:ex1_attacker}. 
The S-S time is $\tau^Y(\y_s,\y_g)=1$ for \attacker{}.

Since the equality in \eqref{eq:population_nontrivial_game} holds, the safe set is a single point: $\x_s=\polyreq(\y_s)$ and $\x_g=\polyreq(\y_g)$.\footnote{We selected such a critical case for the simplicity in the visualization. The same pair $(\y_s,\y_g)$ guarantees \attacker{}'s win in this graph as long as $X<4$.}
The states $\x_s$ and $\x_g$ are shown as blue robots in Fig.~\ref{fig:ex1_attacker}(a) and (b), respectively.
Starting from (a), the required allocation in (b) cannot be achieved by \defender{} within 1 time step, since the resource at node $5$ cannot contribute to any of the required resource at nodes $2, 3$ or $4$ in the immediate next time step.
In fact, Algorithm~\ref{alg:S-S-time} will show that $\tau^X(\x_s,\x_g)=2$ for \defender{}.

This example provides the case where we have a two-step Attacker strategy that guarantees its win.
Specifically, the Attacker first moves all of its resource to node 1 and then move to node 2.

\begin{figure}[t]
\vspace{+5pt}
    \centering
    \includegraphics[width=0.42\textwidth]{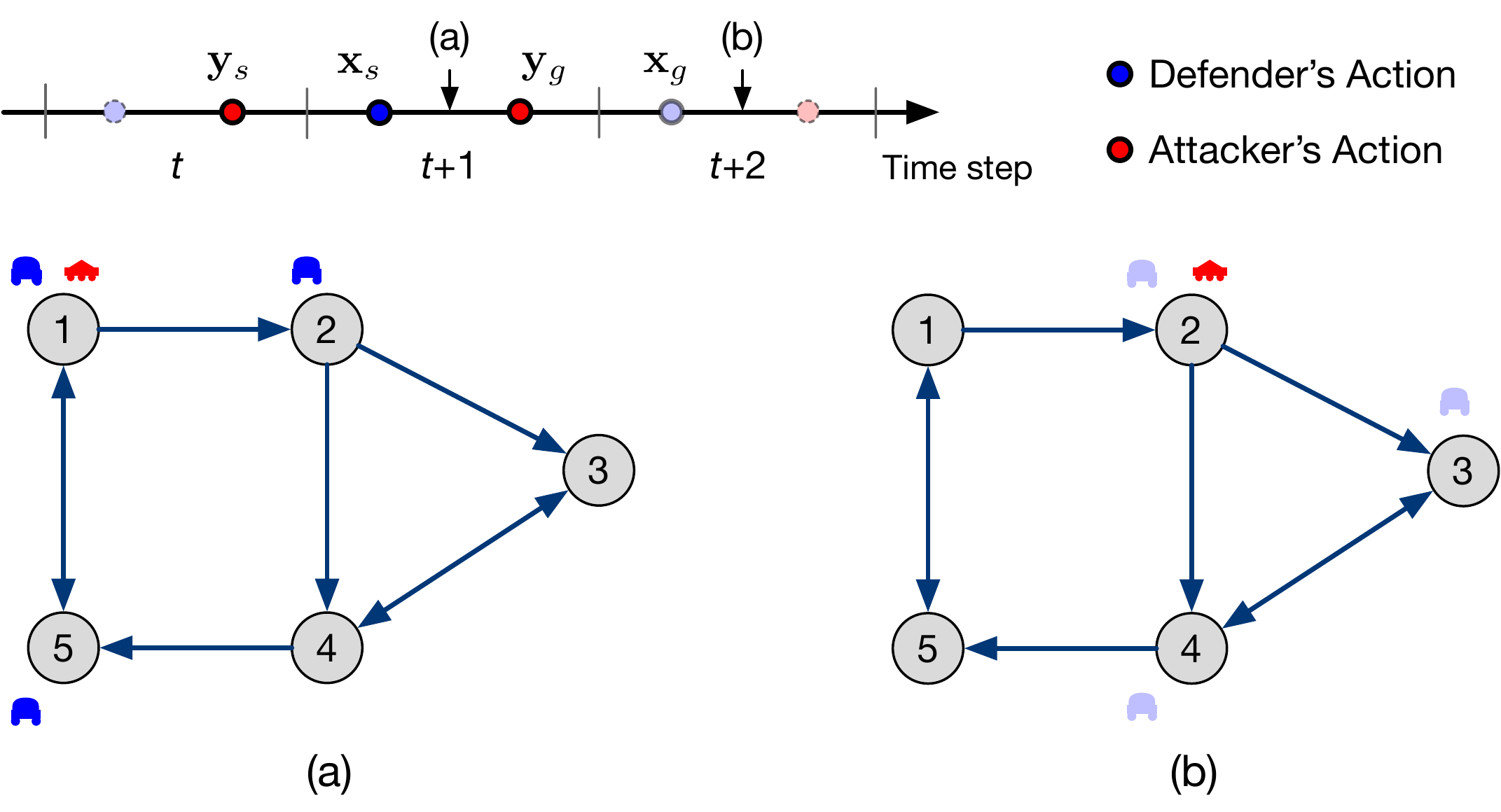}
    \vspace{-10pt}
    \caption{
    An example of \attacker{}-winning graph for $X=3Y$. (Self-loops are omitted for clarity.)
    Starting from the configuration in (a), \attacker{} can move all its resource to node 2 as shown in (b). The blue configuration {in (a)} indicates the necessary and sufficient resources {for \defender{}} to guarantee guarding; the one in (b) is a hypothetical state (shown in light blue) that {is required for guarding at the next time step. However, this hypothetical state} cannot be achieved in one time step from (a).
    }
    \label{fig:ex1_attacker}
\end{figure}

%

\subsection{Attacker Winning Initial Condition}
Finally, we demonstrate the use of Algorithm~\ref{alg:open-loop-attacker}.
Consider the graph in Fig.~\ref{fig:attacker_winning_initial_condition} with total resources $Y= 1$ and $X= 7$.
This is a non-degenerate scenario.
Furthermore, there is no pair $(\y_s,\y_g)$ that qualifies as the solution to Algorithm~\ref{alg:attacker-winning-graph}, which is obvious from the fact that \defender{} can cover each node with one unit of resource.
This setup indicates that the graph, together with the given total resources, does not trivially lead to \attacker{}'s victory.
Therefore, \attacker{} must rely on a favorable initial condition in order to win.
Suppose the initial condition is the one given in Fig.~\ref{fig:attacker_winning_initial_condition}(a), i.e., $\x_0=[4,1,1,1,0,0,0]^\top$ and $\y_0=[0,0,0,1,0,0,0]^\top$.
Given the graph and this initial condition, Algorithm~\ref{alg:open-loop-attacker} identifies a solution $\y_g = [0,0,0,0,1,0,0]^\top$. 
As shown in Fig.~\ref{fig:attacker_winning_initial_condition}, there is no admissible Defender strategy to allocate one unit of resource on both nodes $6$ and $7$ within two time steps from the initial state.
Relating to the reachable sets, the two-step reachable set $\RSet^2(\x_0)$ does not contain any state that has at least one resource on both nodes $6$ and $7$.
Consequently, after observing which of the two nodes \defender{} allocated resource to at time step $1$, \attacker{} can reallocate its resource to the node with less than one Defender resource to terminate the game as shown in Fig.~\ref{fig:attacker_winning_initial_condition}(d). 
\begin{figure}[h]
    \centering
    \includegraphics[width=0.43\textwidth]{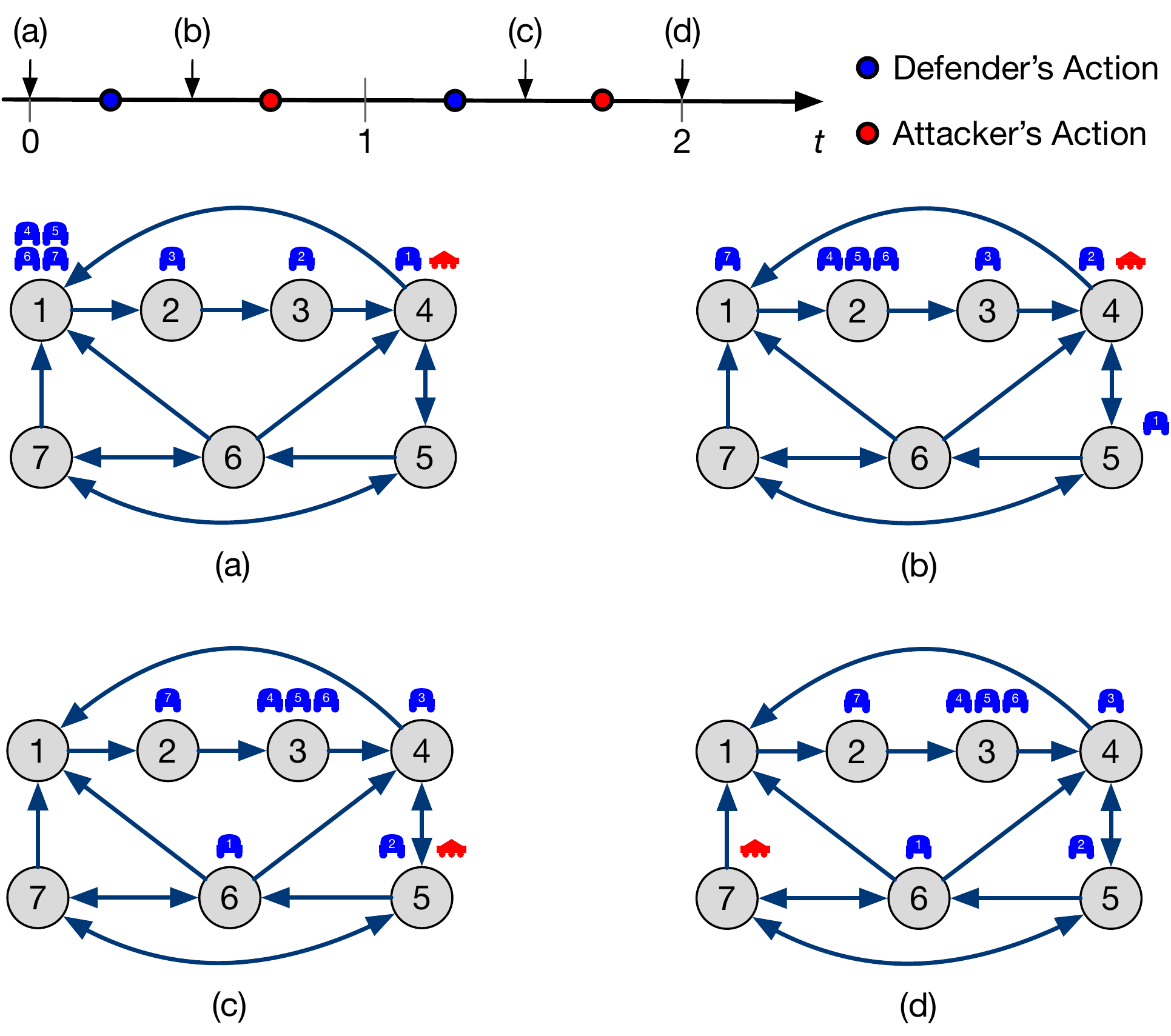}
    \vspace{-5pt}
    \caption{An example of Attacker winning initial condition. All nodes have self-loops. 
    The Defender can defend at time step $1$ as shown in (b).
    When selecting its action $K_1$ at time $1$, \defender{} must decide between relocating the resource at node $5$ in (b) to either node $6$ or node $7$.
    We use the scenario of reallocating to node $6$ in (c).
    After observing this move $K_1$, \attacker{} can then select its action $F_1$ and allocate its resource to node $7$ and terminate the game at time step $2$, as depicted in (d).}
    \label{fig:attacker_winning_initial_condition}
    \vspace{-5pt}
\end{figure}

\section{Conclusion} \label{sec:conclusion}
In this work, we extended the Colonel Blotto game to a dynamic setting, where the locations are modeled as nodes in a graph. 
Instead of achieving a desired allocation instantly, we require the resources of each player to traverse through the edges of the graph. 
Based on the structure in the dynamics, we presented an efficient reachable-set approach to predict the evolution of the game. 
The sufficient winning conditions for the Attacker and the Defender are presented, and we designed algorithms that verify these conditions through the propagation of reachable sets. 
Finally, we demonstrated the efficacy of the proposed approach via some illustrative examples. 
%
Our ongoing work investigates tighter conditions for players' win for broader class of graphs.

\paragraph*{Future work}
There are several directions for future extensions.
Firstly, the assumption of sequential actions can be modified to simultaneous actions, which may lead to mixed strategies.
Secondly, the engagement rules and payoffs can be modified to accommodate stage-based payoffs, changes in the amount of resources, and other terminal conditions.
Finally, the information structure may be relaxed to consider decentralized version of the game where decisions are made at the node level or at the agent level.

\section*{Acknowledgement}
The authors thank Lifeng Zhou, Austin Chen, Panagiotis Tsiotras, and Christopher Kroninger for useful discussions.

\bibliographystyle{IEEEtran}
\bibliography{refs}

\end{document}